\documentclass[floatfix,showpacs,amsmath,amssymb,aip,jmp]{revtex4-1}

\usepackage[USenglish]{babel} 
\usepackage{dsfont}
\usepackage{rotating}
\usepackage{xcolor}
\usepackage{amsmath}
\usepackage{amssymb}
\usepackage{amsthm}
\usepackage[urlcolor=blue,colorlinks=true,linkcolor=blue,pdfstartview={FitH},bookmarks=false]{hyperref} 

\newcommand{\ket}[1]{\ensuremath{|#1\rangle}}
\newcommand{\bra}[1]{\ensuremath{\langle#1|}}

\newcommand{\ie}{\emph{i.e.}}
\newcommand{\etal}{\emph{et al}}

\newcommand{\p}{\scriptscriptstyle{+}}
\newcommand{\m}{\scriptscriptstyle{-}}

\newcommand{\dg}{\dagger}
\newcommand{\mf}{\mathbf}

\newcommand{\mb}{\mathbb}

\newcommand{\txt}[1]{\text{#1}}

\newtheorem{prop}{Proposition}

\begin{document}

\title{Generalized parity in multi-photon Rabi model}

\author{Bart{\l}omiej Gardas}
\address{Institute of Physics, University of Silesia, PL-40-007 Katowice, Poland}
\email{bartek.gardas@gmail.com}

\author{Jerzy Dajka}
\email{jerzy.dajka@us.edu.pl}

\begin{abstract}
Quantum multi--photon spin--boson model is considered. 
We solve an operator Riccati equation associated with 
that model and present a candidate for a generalized 
parity operator allowing to transform spin--boson 
Hamiltonian to a block diagonal form what indicates 
an existence of the related symmetry of the model. 
\end{abstract}

\maketitle

\section{Introduction}
Phenomenological modeling of interacting matter and (quantized) light in physics has a long 
and interesting history~\cite{haken}. It is of quantum optical origin~\cite{vedral} but is 
present in wide range of other branches of physics such as condensed matter~\cite{hol,*n1,*n2,*n3}
or involving mechanical oscillators~\cite{irish,*schwab}. The Rabi model~\cite{rabi1,*rabi2}, 
describing a qubit coupled to a single mode electromagnetic field, is the one which has 
attracted continuous attention for almost a century. More~\cite{braak_letters,*braak_pra,*rabi_exact,*rabi_Nstates} and less~\cite{SB_Fannes,*SB_Spohn} recent studies on its integrability
have inspired increasingly growing research.     

An existence of a symmetry of any quantum model is directly related to a quality of our understanding of
its properties~\cite{wigner}. A `sufficient' (in certain sense) symmetry can result in an integrability
of the model~\cite{perelomov}. That is why seeking for any underlying symmetry of quantum models is 
always of great interest and often of great importance. In this paper we present our contribution
to this activity. We consider a family of generalized single--mode Rabi models~\cite{srwa}: 
\begin{equation}
\label{sb}
\txt{H}=\alpha\sigma_z+\omega a^{\dg}a+\sigma_x\left(g^*a^k+g(a^{\dg})^k\right),
\end{equation}
where $\sigma_z$ and $\sigma_x$ are the Pauli matrices, $\alpha$ and $\omega$ correspond to the energy
gap of the spin and boson, respectively, whereas $a$ and $a^{\dg}$ are the annihilation and creation
operators of quantized mode of light satisfying canonical commutation relation, $[a,a^{\dagger}]=\mathbb{I}$. 
It is assumed that the coupling between the qubit and the field, controlled by the strength constant $g$, 
incorporates $k>0$ photons. 

In this paper, by solving an operator Riccati equation associated with Eq.~(\ref{sb}), we construct an 
operator exhibiting significant similarities to the parity operator acting on the bosonic space. This
operator, the generalized parity, can be used to simplify multi--photon Rabi model~(\ref{sb}) and transform
it to a block--diagonal form. Our work is a complementary expansion of certain results obtained in Ref.~\cite{srwa} for $k=1$ and $k=2$ in the context of approximate methods of solving the Rabi model.  

The paper is organized as follows: In Sec.~\ref{two} we present operator Riccati equation associated 
with~(\ref{sb}) serving as a main tool applied in our studies. Next, in Sec.~\ref{three} the known
results concerning the $k=1$, $2$ cases are reviewed. Sec.~\ref{four} has been devoted to the construction
of the generalized parity and contains main results of our work. Finally, in Sec.~\ref{five}, followed 
by  conclusions, we apply the general construction to a simple example.  

\section{A tool: Riccati equation}
\label{two}

Multi--photon Rabi model considered here belongs to a general class of qubit--environment  composite systems described by Hamiltonian 
\begin{equation}\label{full}
\txt{H}_{\txt{QE}}=\txt{H}_{\txt{Q}}\otimes\mathbb{I}_{\txt{E}}+
                  \mathbb{I}_{\txt{Q}}\otimes\txt{H}_{\txt{E}}+\txt{H}_{\txt{int}}
\sim
\begin{bmatrix}
\txt{H}_{\p} & \txt{V} \\
\txt{V}^{\dg} & \txt{H}_{\m} 
\end{bmatrix}\equiv\mf{H}_{\text{QE}},
\end{equation}
where $\txt{H}_{\txt{Q}}$ ($\txt{H}_{\txt{E}}$) is the Hamiltonian of the system (environment).
$\txt{H}_{\txt{int}}$ is the interaction of the qubit with its  surroundings. $\mathbb{I}_{\txt{Q}}$ and
$\mathbb{I}_{\txt{E}}$ are identities acting on corresponding Hilbert spaces $\mathbb{C}^2$ and 
$\mathcal{H}_{\txt{E}}$. The total Hamiltonian $\txt{H}_{\txt{QE}}$ acts on $\mathbb{C}^2\otimes\mathcal{H}_{\txt{E}}$ and the symbol $\sim$ should be understand as
`\emph{it corresponds to}' in the sense of block operator matrix representation of operators. 
This correspondence  is established via 
the isomorphism $\mathbb{C}^2\otimes\mathcal{H}_{\txt{E}}\sim\mathcal{H}_{\txt{E}}\oplus\mathcal{H}_{\txt{E}}$.
Finally, the form of remaining operators $\txt{H}_\pm$ and  $\txt{V}$
depends upon how $\txt{H}_{\txt{Q}}$, $\txt{H}_{\txt{E}}$ and $\txt{H}_{\txt{int}}$ are defined.

Any steps toward diagonalization of $\mf{H}_\text{QE}$ is valuable as it can be followed by 
variety of different approximation schemes~\cite{srwa}. There is often an additional benefit
emerging form  such transformations which can help to exhibit useful symmetry properties being
often obscured by an `improper choice' of a basis. As it is pointed out below it is also the case
of the multi--photon Rabi model~(\ref{sb}) discussed in this paper. 

Our idea originates from an observation that Hamiltonian $\mf{H}_\text{QE}$ can be converted to a
block--diagonal form 

\begin{equation}
\label{diag}
\mf{S}^{-1}\mf{H}_{\text{QE}}\mf{S} =
       \begin{bmatrix}
    \txt{H}_{\p}+\txt{VX} & 0 \\
    0 & \txt{H}_{\m}- (\txt{VX})^{\dg}
    \end{bmatrix},
    \quad\text{with}\quad
     \mf{S} =
    \begin{bmatrix}
    \mathbb{I}_{\txt{E}} & -\txt{X}^{\dg} \\
    \txt{X} & \mathbb{I}_{\txt{E}}
    \end{bmatrix}, 
\end{equation}
provided that $\txt{X}$ satisfies an operator Riccati equation
\begin{equation}
\label{mgricc}
\txt{XVX}+\txt{X}\txt{H}_{\p}-\txt{H}_{\m}\txt{X}-\txt{V}^{\dg}=0.
\end{equation}
For general considerations regarding an operator Riccati equation we refer the reader 
to~\cite{Vadim,*RiccEq,egorov}. This equation provides valuable tool allowing to study
the exact diagonalization~\cite{Rdiag,*gardas4}, stationary states~\cite{gardasPuchala}
and in general, the dynamics~\cite{gardas_pra} of two level open quantum systems~\cite{breuer_book,*alicki}.
From  the decomposition~(\ref{diag}) it is evident that the dynamics of a qubit--environment
quantum system is actually governed by the Riccati~(\ref{mgricc}) and pair of 
\emph{uncoupled} Schr\"{o}dinger equations.

For the $k$-photon Rabi model studied in our paper
\begin{equation}
\label{krabih}
\txt{H}_{\pm}=\omega a^{\dg}a\pm \left(g^*a^{k}+g(a^{\dg})^k\right),
\quad
\txt{V}=\alpha\mb{I}_{\mathcal{H}_{\txt{B}}} 
\end{equation}
and the corresponding Riccati equation reads as follows:
\begin{equation}
\label{gricc}
\alpha\txt{X}^2+\txt{X}\txt{H}_{\p}-\txt{H}_{\m}\txt{X}-\alpha=0.
\end{equation}
Its mathematical properties has already been addressed in literature~\cite{gardas,*gardas2}. 
In Eq.~(\ref{gricc}), $\txt{H}_{\pm}$ are operators acting on the bosonic Fock space $\mathcal{H}_{\txt{B}}$,  $\alpha$ is a real constant, whereas $\txt{X}$ is a solution to be found. If it does not lead to a confusion, we  write  $\alpha$ rather than $\alpha \mb{I}_{\mathcal{H}_{\txt{B}}}$, with $\mb{I}_{\mathcal{H}_{\txt{B}}}$ being the identity on $\mathcal{H}_{\txt{B}}$. 
\section{Known solutions: $k=1,2$.}
\label{three}
For the sake of self-consistency, we begin with reviewing known solutions and their properties 
for the two particular cases, where $k=1$, $2$. For the simplest possible case, $k=1$ the solution
of the Riccati equation~(\ref{gricc}) was found in~\cite{gardas4} to be the bosonic parity operator 
\begin{equation}
\label{parity}
 \txt{P} = \sum_{n\in\mb{N}}e^{i\pi n}\ket{n}\bra{n}
               = \sum_{n\in\mb{N}}(-1)^n\ket{n}\bra{n},
\end{equation}
which can also be written in a more compact form as $\txt{P}=\exp(i\pi a^{\dg}a)$, where $\{\ket{n}\}_{n\in\mathbb{N}}$ is the Fock basis, \ie, $a^{\dg}a\ket{n}=n\ket{n}$. Such 
operator is both hermitian and unitary, hence it is an involution 
($\txt{P}^2=\mathbb{I}_{\mathcal{H}_{\txt{B}}}$). Interestingly, it solves Eq.~(\ref{gricc})
for both $\alpha=0$ (dephasing~\cite{SB_Alicki,*dajka_cat}) and  $\alpha\not=0$ (exchange 
energy between the systems is present) cases, although they reflect quite different
physical processes. 

In the context of RWA--type approximation methods the two-photon ($k=2$) Rabi model was studied 
in details within~\cite{srwa}. The two--photon parity operator 
\begin{equation}
 \label{parity2}
 \txt{T}:=\exp\left[i\frac{\pi}{2}a^{\dg}a\left(a^{\dg}a-1\right)\right],
\end{equation} 
was introduced therein. It has not been stated explicitly in~\cite{srwa} but the parity 
$\text{T}$ is, as will be shown below, a solution of the Riccati equation~(\ref{gricc}) for $k=2$.  

\section{General case: $k>0$}
\label{four}

In what follows we show how to construct a  solution of the Riccati equation~(\ref{gricc}) with
coefficients $\txt{H}_{\pm}$ provided by~(\ref{krabih}) in the general case $k>0$. Before we start
let us emphasize that the parity operator $\txt{P}$ ($\txt{T}$) introduced in the preceding section
solves Eq.~(\ref{gricc}) not only for $k=1$ ($k=2$) but also for all odd $k=2n+1$ (even, of the form
$k=2n+4$) cases. This has already been noticed in~\cite{srwa}. Here we will not only fill the 
remaining gap $k=2n+2$ but also present unified approach allowing to obtain a linear solution for 
arbitrary $k$. As a first step toward constructing this solution, we define a family of orthogonal
projectors 
\begin{equation}
\label{family}
\txt{P}_l:=\sum_{n=0}^{\infty}\ket{n,l}\bra{n,l},
\quad\txt{with}\quad
\ket{n,l}:=\ket{kn+l-1},
\end{equation}
for $n\in\mathbb{N}$ and $1\le l\le k$. The states $\ket{n,l}$ satisfy the following orthogonality condition:
\begin{equation}
\label{delta}
\langle{i,l}|j,l\rangle=\delta_{kn+i-1,km+j-1}=\delta_{ij}\delta_{nm},
\end{equation}
where $\delta_{xy}$ is the Kronecker delta. The first equality in Eq.~(\ref{delta})
comes from the orthogonality of the Fock basis. The second one can be justified as 
follows. 

When $i=j$ both sides of~(\ref{delta}) reduce  to $\delta_{nm}$ since 
$\delta_{kn+i-1,km+i-1}=\delta_{nm}$. If $i\not=j$ (say $i>j$) the right hand side is
zero. The left hand side also vanishes, as one gets either $m=n$ or $m\not=n$ in this 
case. Indeed, if $m=n$  then, to get nonzero left hand side, one would expect $i-j=0$ 
what is impossible. Finally, for $m\not=n$ ($m>n$, say) one would expect that $k(m-n)=j-i$
(in order to keep the left side nonzero) what also does not occur as $k>i-j$ and it cannot
divide $i-j$.  

For a given family of orthogonal projectors one can split the space $\mathcal{H}_{\txt{B}}$
into $k$ subspaces so that
\begin{equation}
\label{split}
 \mathcal{H}_{\txt{B}}=
 \mathcal{H}_{1}\oplus\mathcal{H}_{2}\oplus\cdots\oplus\mathcal{H}_{k-1}\oplus\mathcal{H}_{k}=
 \bigoplus_{l=1}^{k}\mathcal{H}_{l},
\end{equation}
where  $\mathcal{H}_{l}:=\txt{P}_l\left(\mathcal{H}_{\txt{B}}\right)$. 
The symbol $\oplus$ indicates the (orthogonal) direct sum of Hilbert spaces. Hereafter, we use 
it interchangeably with $+$ when it refers to the sum of operators.

The decomposition~(\ref{split}) allows us to think of operators $\txt{H}_{\pm}$ as of $k\times k$
block operator matrices $[\txt{H}_{\pm}^{lm}]_{k\times k}$ such that
\begin{equation}
\txt{H}_{\pm}^{lm}:=\txt{P}_l(\txt{H}_{\pm})\txt{P}_m:\mathcal{H}_{l}\rightarrow\mathcal{H}_{m}. 
\end{equation} 
Obviously, $\txt{H}_{\pm}^{ll}$ act within the space $\mathcal{H}_{l}$, and
they may be considered as  operators $\txt{H}_{\pm}$ restricted to the space $\mathcal{H}_{l}$, \ie,
$\txt{H}_{\pm}^{ll}:=(\txt{H}_{\pm})_{|\mathcal{H}_l}$. The off-diagonal elements
$\txt{H}_{\pm}^{lm}$ act between the subspaces $\mathcal{H}_{l}$, $\mathcal{H}_{m}$ and therefore 
transform state from one space into the other. In this simple picture instead of speaking of compositions, 
sums and any other operations involving two or more operators acting on $\mathcal{H}_{\txt{B}}$, we 
operate with corresponding  $k\times k$ matrices.
The block operator matrix representation of $\txt{H}_{\pm}$ is useful provided that it has a relatively
simple form. As an example, an occurrence of a block diagonal structure ($\txt{H}_{\pm}^{lm}=\delta_{lm}\txt{H}_{\pm}^{(l)}$)
would be an ideal situation. As it will be shortly seen this is indeed the case here. 

Since~(\ref{delta}) holds, it is a matter of straightforward calculations to show that $\bra{i,l}a^{\dg}a\ket{j,m}=\xi_{jm}\delta_{ij}\delta_{lm}$ and
$\bra{i,l}a^k\ket{j,m}=\eta_{jm}\delta_{i,j-1}\delta_{lm}$, for a given $k$, where
\begin{equation}
\xi_{jm}=kj+m-1 
\quad\txt{and}\quad
\eta_{jm}=\sqrt{\frac{(kj+m-1)!}{[k(j-1)+m-1]!}}.
\end{equation}
By making use of the above equations we obtain
\begin{equation}
\begin{split}
\text{P}_l(a^{\dg}a)\txt{P}_m
           &=
          \sum_{i,j=0}^{\infty}\left(\bra{i,l}a^{\dg}a\ket{j,m}\right)
          \ket{i,l}\bra{j,m}  
          = \left(\sum_{j=0}^{\infty}\xi_{jm}\ket{j,l}\bra{j,m}\right)\delta_{lm}\\
          &\equiv \left(k\txt{n}_l+(l-1)\mb{I}_{\mathcal{H}_l}\right)\delta_{lm} 
          \equiv
 \text{N}_l\delta_{lm}, 
\end{split}
\end{equation}
where we have introduced $\txt{N}_l$ - the number operator restricted to the subspace
$\mathcal{H}_l$, $\txt{N}_{l}:=(a^\dg a)_{|\mathcal{H}_l}$:
\begin{equation}
 \txt{N}_{l}=k\txt{n}_l+(l-1)\mb{I}_{\mathcal{H}_l},
 \quad
 \txt{n}_l=\sum_{n=0}^{\infty}n\ket{n,l}\bra{n,l},
\end{equation}
with $\mb{I}_{\mathcal{H}_l}$ being the identity on $\mathcal{H}_l$. In a similar fashion we have
\begin{equation}
 \text{P}_l(a^k)\txt{P}_m 
           =%
          \sum_{i,j=0}^{\infty}\left(\bra{i,l}a^k\ket{j,m}\right)
          \ket{i,l}\bra{j,m}\\ 
          = \left(\sum_{j=0}^{\infty}\eta_{jm}\ket{j-1,l}\bra{j,m}\right)\delta_{lm},
\end{equation}
from which  $\txt{A}_l:=\sum_{j=0}^{\infty}\eta_{jl}\ket{j-1,l}\bra{j,l}$  is
nothing but $(a^k)_{|\mathcal{H}_l}$ and plays on $\mathcal{H}_l$ a role of annihilation operator.

Combining all the above results  into a single equation we finally obtain the block matrix
representation of $\txt{H}_{\pm}$ with respect to the decomposition~(\ref{split}). It 
has a diagonal structure indeed: 
\begin{equation}
\label{bds}
 \txt{H}_{\pm}^{lm}=\left[\omega\txt{N}_l\pm(g^*\txt{A}_l+g\txt{A}_l^{\dg})\right]\delta_{lm}
 \equiv \txt{H}_{\pm}^{(l)}\delta_{lm},
\end{equation}
for $ l,m\leq k$.
 
Having~(\ref{bds}) in place, we can formulate the result of this paper:

\begin{prop} 
The solution to the Riccati equation~(\ref{gricc}) is given by the {\it generalized parity operator}
\begin{equation}
\label{sol}
 \txt{X}_k=\sum_{l=1}^{k}\sum_{n=0}^{\infty}(-1)^n\ket{n,l}\bra{n,l}, \quad\txt{for}\quad k>0.
\end{equation}
%
\end{prop}

\begin{proof}

We begin by defining {\it partial parities} $\txt{J}_l$ to be
\begin{equation}
\label{ps}
 \txt{J}_l:=e^{i\pi n_l}=\sum_{n=0}^{\infty}(-1)^n\ket{n,l}\bra{n,l},
 \quad
 l\leq k.
\end{equation}
It is justified to refer to $\txt{J}_l$ as a parity since it possesses 
all the desired properties  required from the parity operator on
$\mathcal{H}_l$, namely
\begin{equation}
\label{par}
\txt{J}_l^2=\mb{I}_{\mathcal{H}_l},
\quad
[\txt{N}_l,\txt{J}_l]=0,
\quad\txt{and}\quad
\txt{J}_l\txt{A}_l\txt{J}_l=-\txt{A}_l.
\end{equation}
As an immediate consequence of this conditions \mbox{$\txt{J}_l\txt{H}_{\p}^{(l)}\txt{J}_l=\txt{H}^{(l)}_{\m}$},
what in the block operator matrix terminology developed in the preceding section leads to
\begin{equation}
\label{hpm}
\txt{X}_k\txt{H}_{\p}\txt{X}_k
\sim
\txt{diag}[\txt{J}_1\txt{H}_{\p}^{(1)}\txt{J}_1,\dots,\txt{J}_k\txt{H}_{\p}^{(k)}\txt{J}_k]
\sim \txt{H}_{\m},
\end{equation}
or simply $\txt{X}_k\txt{H}_{\p}=\txt{H}_{\m}\txt{X}_k$. Therefore, to prove~(\ref{gricc}) holds true, it 
is sufficient to show that $\txt{X}_k$ is an involution. It can be established in a following way:
\begin{equation}
\label{inv}
 \txt{X}_k^2
    \sim\txt{diag}[\txt{J}_1^2,\dots,\txt{J}_k^2]
    \sim\bigoplus_{l=1}^{k}\mathbb{I}_{\mathcal{H}_{l}}
    =\mathbb{I}_{\mathcal{H}_{\txt{B}}}.
\end{equation}
Both in Eqs.~(\ref{hpm}) and~(\ref{inv}) we have taken into account
the correspondence $\txt{X}_k\sim\txt{diag}[\txt{J}_1,\dots,\txt{J}_k].\mbox{\qedhere}$
\end{proof}

\section{Examples}
\label{five}

It is interesting to see how the solutions which  have been found in the previous section  
can be used to recover the known results for $k=1$, $2$.
When $k=1$ there is only one subspace of $\mathcal{H}_{\txt{B}}$, namely $\mathcal{H}_{\txt{B}}$
itself, and  $\txt{X}_1$ is equal to the bosonic parity operator~(\ref{parity}). 
For $k=2$ there are only two projection within the family of operators~(\ref{family}), \ie,
\begin{equation}
\txt{P}_{1}=\sum_{n=0}^{\infty}\ket{2n}\bra{2n},
\quad
\txt{P}_{2}=\sum_{n=0}^{\infty}\ket{2n+1}\bra{2n+1},
\end{equation}
or in a compact form $\txt{P}_{1,2}=\tfrac{1}{2}(\mathbb{I}_{\mathcal{H}_{\txt{B}}}\pm\txt{P})$, which
according to~(\ref{split}) split the bosonic Hilbert space into two subspaces. The first one consists 
only of odd, while the second one of even Fock states: 
\begin{equation}
\label{eo}
\mathcal{H}_{\txt{B}}=
\txt{span}\{\ket{2n}:n\in\mathbb{N}\}\oplus
\txt{span}\{\ket{2n+1}:n\in\mathbb{N}\}.
\end{equation}
The block operator matrix representation of $\txt{H}_{\pm}$ with respect to~(\ref{eo}) reads
\begin{equation}
\label{tbt}
\txt{H}_{\pm}
\sim
\begin{bmatrix}
\txt{H}_{\pm}^{(1)} & 0 \\
0 & \txt{H}_{\pm}^{(2)}
\end{bmatrix}
\\
=
\begin{bmatrix}
\txt{J}_1 & 0 \\
0 & \txt{J}_2
\end{bmatrix}
\begin{bmatrix}
\txt{H}_{\mp}^{(1)} & 0 \\
0 & \txt{H}_{\mp}^{(2)}
\end{bmatrix}
\begin{bmatrix}
\txt{J}_1 & 0 \\
0 & \txt{J}_2
\end{bmatrix}
\sim
\txt{H}_{\mp},
\end{equation}
where diagonal entries are explicitly given by~(\ref{bds}). To see that the two-photon 
parity $\txt{T}$ is indeed the same as $\txt{X}_2$ one only needs to invoke a simple fact,  
\mbox{$(-1)^{n(2n-1)}=(-1)^{n(2n+1)}=(-1)^{n}$}, then

\begin{equation}
 \label{check}
\txt{T}  =  \sum_{n=0}^{\infty}(-1)^{\frac{n(n-1)}{2}}\ket{n}\bra{n} 
         =  \sum_{n=0}^{\infty}(-1)^n\left(\ket{2n}\bra{2n} +\ket{2n+1}\bra{2n+1}\right).
\end{equation}

Eq.~(\ref{tbt}) also allows us to identify $\txt{X}_2\sim\txt{diag}[\txt{J}_1,\txt{J}_2]$.  

\section{Conclusion}
\label{six}
In summary, we have found the solution of the operator Riccati equation~(\ref{gricc}) associated 
with the multi--photon Rabi model~(\ref{sb}) indicating certain symmetry of the original model. 
This solution is a natural candidate for a parity operator as it not only reduces to well known
\mbox{one--\cite{gardas4}} and two--photon~\cite{srwa} parities but also its properties are that
of a typical parity operator. 

We have also proved by an explicit construction that the solution of~(\ref{gricc}) with coefficients
given by~(\ref{krabih}) exist for every $k>0$. We have excluded from our analysis the case where $k=0$
only because it does not reflect any relevant physical phenomena. Since there is no interaction between
the systems they evolve in time separately. 

The symmetry associated with the generalized parity Eq.~(\ref{sol}) allows to convert Rabi model into a 
block--diagonal form. Such a transformation decouples the original qubit--boson eigenproblem into a pair
of bosonic Schr\"{o}dinger equations. Despite certain mathematical  subtleties of the multi--photon Rabi 
models~\cite{ill} we hope that the results presented in this paper can serve as a starting point 
for useful approximation methods. 

Finally, let us notice that any involution $\txt{J}$ for which $\txt{J}\txt{H}_{\p}=\txt{H}_{\m}\txt{J}$ is a solution of~(\ref{gricc}) for $k>0$. Unfortunately, the question whether all the solution to~(\ref{gricc}) 
are of that special kind still remains open.


\section{Acknowledgments}
This work was supported by the Polish Ministry of Science and Higher Education under project {\bf{Iuventus Plus}}, 
No. 0135/IP3/2011/71 (B. G) and NCN Grant N202 052940
(J. D)

%
\end{document}